\documentclass[12pt]{article}

\usepackage{amsmath,amsthm,amssymb}
\usepackage{mathtools}
\usepackage{bm}
\usepackage{natbib}
\usepackage{geometry}
\usepackage{hyperref}
\usepackage{booktabs}
\usepackage{array}
\usepackage{multirow}
\usepackage{setspace}
\usepackage[protrusion=true,expansion=false]{microtype}
\usepackage{xcolor}
\usepackage{enumitem}

\newtheorem{lemma}{Lemma}
\newtheorem{corollary}{Corollary}
\newtheorem{remark}{Remark}
\theoremstyle{definition}

\newtheorem*{example*}{Example}

\newcommand{\bw}{\bm{w}}
\newcommand{\by}{\bm{y}}

\newcommand{\bG}{\bm{G}}
\newcommand{\blam}{\bm{\lambda}}
\newcommand{\bbeta}{\bm{\beta}}

\newcommand{\A}{\mathcal{A}}
\newcommand{\hatT}{\hat{\bm{T}}}
\newcommand{\hatTHT}{\hat{\bm{T}}^{\mathrm{HT}}}
\newcommand{\CI}{\mathrm{CI}}

\usepackage{float}   

\hypersetup{
  colorlinks=true,
  linkcolor=blue!60!black,
  citecolor=blue!60!black,
  urlcolor=blue!60!black
}

\begin{document}

\title{\textbf{Post-Hoc Inference of Cross-Classified Statistics from
       Hierarchical Bayes Survey Weights}}

\author{Siu-Ming Tam\\
        \small Tam Data Advisory, Australia\\
        \small \texttt{stattam@gmail.com}}

\date{Version~27 --- \today}
\maketitle
\thispagestyle{empty}

\begin{abstract}
\citet{tam2026more} shows that combining Bethel multivariate allocation with Hierarchical Bayes (HB) small area models can substantially reduce survey sample sizes while maintaining domain-level precision and near-nominal coverage of posterior credible intervals (CrIs).

This paper extends that framework to cross-classified statistics derived from HB-calibrated unit record data. Its central contribution is a Post-Hoc Inference Engine (PHIE) that propagates uncertainty from HB domain posterior draws to arbitrary cross-tabulations. PHIE transforms each MCMC draw via chi-square calibration to produce replicate survey weights, from which CrIs are obtained.

Three tiers of statistics are identified. Tier 1-E cells reproduce calibration totals and yield exact posterior CrIs. Tier 2 cells involve filtered sums of calibration variables; PHIE alone undercovers, but a Calibrated Bayes interval (CBI), augmenting PHIE with design-based compositional variance, restores near-nominal coverage. Tier 3-NCV cells involve non-calibration variables; a ratio-based CBI linked to a correlated calibration variable achieves reliable coverage even under weak correlation.

A key empirical finding is that uncertainty in cross-tabulations is driven primarily by compositional sampling variability rather than HB model uncertainty. Resulting CBI-based coefficients of variation remain within standard publication thresholds.

\end{abstract}

\newpage
\setcounter{page}{1}

\section{Introduction}

\citet{tam2026more} showed that a survey sample can be reduced substantially by combining Bethel multivariate allocation with Hierarchical Bayes (HB) small area modelling.  In that framework, the HB model produces posterior totals for a set of target variables, and the final mini-max sample is chosen so that national and domain precision requirements continue to be met.  That earlier paper deals with the design and model-based estimation problem.

The present paper deals with the next operational problem.  Once the reduced survey has been fielded, an NSO will not publish only the target variables used in the HB model.  It will usually release a weighted unit record file and produce many cross-classified tables: for example, employment by age group, employment by occupation, hours worked by domain, income by employment status, or health status by occupation.  These tables are often constructed after the survey design and model fitting have been completed.  They may involve variables that were not part of the HB calibration system.  The question is therefore not how to reduce the sample size, but how to make valid statistical inference for cross-classified statistics derived from an HB-calibrated unit record file.

This is the central problem addressed in this paper.  Point estimation is relatively straightforward: one can construct survey weights so that the weighted sums of the HB calibration variables agree with the corresponding HB posterior estimates.  The more difficult issue is uncertainty.  If the HB posterior uncertainty is not propagated to the derived cross tabulations, the published table may look precise simply because it is based on a calibrated point estimate.  This would understate the true uncertainty.  Conversely, if the posterior draws are propagated mechanically without recognising the relationship between the cross-classified cell and the HB calibration variables, the resulting credible intervals may have poor repeated-sampling coverage.

The paper introduces the Post-Hoc Inference Engine (PHIE) to address this problem.  For each MCMC draw of the HB posterior domain totals, PHIE performs a chi-square calibration of the original survey weights to that draw.  This produces a corresponding set of posterior-calibrated replicate weights.  Applying these replicate weights to any cross-classified statistic gives an empirical posterior distribution for that statistic.  The empirical quantiles of this distribution provide a credible interval when the statistic is directly supported by the HB calibration system, and a quasi-posterior interval otherwise.

A key contribution of the paper is to show that not all cross tabulations have the same inferential status.  The paper therefore develops a three-tier taxonomy.  Tier~1-E cells exactly reproduce one of the HB posterior calibration totals and therefore inherit exact posterior credible intervals.  Tier~2 cells sum a calibration variable over a filter.  These cells are closer to the calibration system, but the uncorrected PHIE interval can still undercover; a CBI is therefore developed to restore repeated-sampling coverage.  Tier~3-NCV cells involve a non-calibration outcome variable.  For these cells, the paper proposes a ratio-estimator CBI, linking the non-calibration variable to the calibration variable most strongly associated with it.

The contribution of the paper is thus not the two-stage sample reduction method itself.  That is the subject of \citet{tam2026more}.  The contribution here is the post-hoc inference machinery needed after such a survey has been conducted: posterior-calibrated replicate weights, credible intervals for calibration-supported cross tabulations, CBI corrections for partially supported cross tabulations, and a ratio-estimator CBI for non-calibration variables.  The methods are illustrated using Australian 2021 Census microdata, where population truth is available and Monte Carlo coverage can be assessed directly.

The remainder of the paper is organised as follows.  Section~\ref{sec:framework} summarises the HB survey framework from \citet{tam2026more} and establishes notation.  Section~\ref{sec:phie} develops the PHIE methodology, proves the calibrated-weight and replicate-weight propagation results, and introduces the tier taxonomy.  Section~\ref{sec:results} presents empirical results for all cell types.  Section 5 concludes with recommendations for NSO publication practice.

\section{The HB Survey Framework}
\label{sec:framework}

\subsection{Setting and Notation}

Let $\mathcal{U} = \{1,\ldots,N\}$ denote a finite population of $N$
units partitioned into $H$ non-overlapping strata indexed
$h=1,\ldots,H$, with stratum sizes $N_1,\ldots,N_H$.  A stratified
simple random sample of size $n_h$ is drawn without replacement from
stratum $h$, giving total sample size $n = \sum_h n_h$.  The sample is
denoted $\A$, with Horvitz--Thompson (HT) weights $w_i = N_h/n_h$ for
unit $i$ in stratum $h$.  The population is partitioned into $D$
geographic domains $d=1,\ldots,D$.  There are $V$ target variables
$v=1,\ldots,V$, each with a domain-level precision target
$\mathrm{CV} \leq g_{d,v}$.

For each record $i \in \A$, define the \textbf{calibration design
vector} by stacking $V$ blocks of dimension $D$:
\[
  \by_i
  = \bigl(y_i^{(v,d)}\bigr)_{v=1,\ldots,V;\; d=1,\ldots,D}
  \;\in\; \mathbb{R}^p, \qquad p = V \times D,
\]
where $y_i^{(v,d)} = y_i^{(v)} \cdot \mathbf{1}(i \in d)$ is the
scalar value of target variable $v$ for unit $i$ if $i$ belongs to
domain $d$, and zero otherwise.  Here $y_i^{(v)} \in \mathbb{R}$
denotes the raw scalar measurement of variable $v$ for unit $i$
without domain interaction; since each unit belongs to exactly one
domain, exactly one entry in each $v$-block of $\by_i$ is non-zero.
In the numerical example of Section 4.1, $V=3$, $D=8$,
$p=24$.

\begin{example*}
To make the definition concrete, consider a person $i$ who resides in
NSW ($d=1$), is employed ($y_i^{(1)}=1$), not unemployed
($y_i^{(2)}=0$), and works $38$ hours per week ($y_i^{(3)}=38$).
The calibration design vector $\by_i \in \mathbb{R}^{24}$ has three
non-zero entries:
\[
  \by_i =
  \underbrace{(1,\, 0,\, 0,\, 0,\, 0,\, 0,\, 0,\, 0)^\top}_{\text{Var1 block}\;\in\;\mathbb{R}^8}
  \!\!\oplus\!\!
  \underbrace{(0,\, 0,\, 0,\, 0,\, 0,\, 0,\, 0,\, 0)^\top}_{\text{Var2 block}\;\in\;\mathbb{R}^8}
  \!\!\oplus\!\!
  \underbrace{(38,\, 0,\, 0,\, 0,\, 0,\, 0,\, 0,\, 0)^\top}_{\text{Var3 block}\;\in\;\mathbb{R}^8},
\]
where within each $V$-block the entry corresponding to the person's
domain ($d=1$, NSW) carries the variable value and all other domain
entries are zero.  
\end{example*}

\subsection{Two-Stage Strategy: Bethel Allocation and HB Reduction}

\citet{tam2026more} establishes a two-stage strategy to reduce the
operating cost of surveys without sacrificing precision.
\textbf{Stage~1} uses the R2BEAT package \citep{falorsi2021} to solve
the Bethel multivariate constrained optimisation \citep{bethel1989},
finding the minimum total sample $n^*$ satisfying national CV targets
of $3\%$ and domain CV targets of $8\%$ for all $V \times D$
variable--domain combinations simultaneously,
replacing the common NSO practice of computing Neyman allocations
separately per variable and taking the element-wise maximum --- a
procedure that both inflates cost and fails to guarantee domain
precision \citep{bethel1989}.  \textbf{Stage~2} fits HB models to a
nested sub-sample of size $n < n^*$, exploiting the
strength-borrowing properties of HB to satisfy all precision targets
at the smaller sample.  The optimal reduction fraction is determined
by a four-gate eligibility algorithm covering CV, convergence,
national accuracy, and domain accuracy.

In the synthetic Labour Force Survey simulation of \citet{tam2026more},
the Bethel benchmark is $n^* = 91{,}308$ and the HB-combined sample is
$n^{\mathrm{HB}} = 18{,}262$ (an $80\%$ reduction).  Monte Carlo
results with $B=5{,}000$ replications confirm CI coverage of $93.0\%$
(Employment), $97.7\%$ (Unemployment), $100.0\%$ (Hours Worked), with
CV gate pass rates exceeding $95\%$ for all three variables.

\subsection{HB Models}

For binary variables (e.g. Employment, Unemployment), the logit-normal
binomial HB model of \citet{rao2015} is used:
\begin{equation}\label{eq:binary_hb}
  m_h \mid p_h \;\sim\; \mathrm{Binomial}(n_h, p_h), \qquad
  \mathrm{logit}(p_h) = \bm{z}_h^\top \bbeta + v_h, \qquad
  v_h \;\overset{\mathrm{iid}}{\sim}\; \mathcal{N}(0,\sigma_v^2),
\end{equation}
where $m_h = \sum_{i \in \A_h} y_i^{(v)}$ is the observed count of
successes for target variable $v$ in stratum $h$, and $\bm{z}_h$ is a vector of
stratum-level auxiliary covariates.
The prior on the random-effects variance is
$\sigma_v^2 \sim \mathrm{Inv}\text{-}\chi^2(\nu, s^2)$, with
hyperparameters $(\nu, s^2)$ calibrated via a grid search that
maximises CI coverage while minimising domain maximum average relative
error (MARE) \citep{tam2026more}.

For the continuous variable (e.g. Hours Worked), the Gaussian
Fay--Herriot model \citep{fayherriot1979} is used:
\begin{equation}\label{eq:gaussian_hb}
  \hat{\theta}_h \mid \theta_h \;\sim\; \mathcal{N}(\theta_h, \psi_h),
  \qquad
  \theta_h = \bm{z}_h^\top \bbeta + v_h, \qquad
  v_h \;\overset{\mathrm{iid}}{\sim}\; \mathcal{N}(0,\sigma_v^2),
\end{equation}
where the sampling variance $\psi_h = \mathrm{DEFF}_h \cdot
(1-f_h)S_h^2/n_h$ is treated as known, incorporating the design
effect, the finite population correction, and the within-stratum
sample variance.  Both models are estimated using the R package
\texttt{mcmcsae} \citep{boonstra2021}.  Three independent MCMC chains
are run; convergence is assessed via the Gelman--Rubin statistic
\citep{gelman1992}.  In the simulated application of \citet{tam2026more},
$\hat{R}_{\max} = 1.002$ across all three variables, indicating
excellent convergence.

\section{The Post-Hoc Inference Engine}
\label{sec:phie}

\subsection{Setup}
\label{sec:setup}

Let $\bm{T} = (T^{(v,d)})_{v,d} \in \mathbb{R}^p$ denote the random
$p$-vector of domain totals with posterior distribution
$\Pi(\bm{T} \mid \text{data})$ arising from the HB model,
$\mathbb{E}[\bm{T} \mid \text{data}]
        \;\in\; \mathbb{R}^p$
be the vector of HB posterior means, 
$\hatT = \bigl(\hat{T}^{(v,d)}\bigr)_{v,d}$ denote their estimates, 
$\{\hatT^{(b)}\}_{b=1}^{B}$ denote $B$ MCMC draws from
$\Pi(\bm{T} \mid \text{data})$, and $\hatTHT = \sum_{i \in \A} w_i
\by_i$ denote the corresponding vector of HT estimates.

Let $\mathbf{Y}$ denote the $n \times p$ calibration design
matrix whose $i$-th row is the transpose of the calibration design vector, $\by_i^\top$, :
\[
  \mathbf{Y}
  = \begin{pmatrix}
      \by_1^\top \\ \by_2^\top \\ \vdots \\ \by_n^\top
    \end{pmatrix}
  \in \mathbb{R}^{n \times p},
\]
and let $\mathbf{W} = \mathrm{diag}(w_1, \ldots, w_n)$ be the
diagonal matrix of Horvitz--Thompson weights.  Define the
\textbf{calibration instrument cross-product matrix}
\begin{equation}\label{eq:G}
  \bG
  = \mathbf{Y}^\top \mathbf{W} \mathbf{Y}
  = \bigl(\mathbf{W}^{1/2}\mathbf{Y}\bigr)^\top
    \bigl(\mathbf{W}^{1/2}\mathbf{Y}\bigr)
  = \sum_{i \in \A} w_i \by_i \by_i^\top
  \;\in\; \mathbb{R}^{p \times p}.
\end{equation}

\begin{remark}[Rank of $\bG$]\label{rem:rank}
Since $\bG = (\mathbf{W}^{1/2}\mathbf{Y})^\top
(\mathbf{W}^{1/2}\mathbf{Y})$ is a Gram matrix, it is symmetric
positive semi-definite for any $\mathbf{Y}$ and any $\mathbf{W}$ with
positive diagonal entries.  Its rank satisfies
\[
  \mathrm{rank}(\bG)
  \;=\; \mathrm{rank}\!\bigl(\mathbf{W}^{1/2}\mathbf{Y}\bigr)
  \;=\; \mathrm{rank}(\mathbf{Y}),
\]
the last equality holding because $\mathbf{W}^{1/2}$ is an
$n \times n$ diagonal matrix with strictly positive diagonal entries.
Consequently $\bG$ is positive definite --- and $\bG^{-1}$ exists ---
if and only if $\mathbf{Y}$ has rank $p$, which requires
at least $p$ linearly independent calibration design vectors among the
$n$ sampled records.  In a survey setting where $n \gg p$, this condition is almost always
satisfied in practice.  We make this assumption throughout the sequel.
\end{remark}

\subsection{Posterior-Calibrated Weights and CrIs for Cross Classified Statistics}

\begin{lemma}[Posterior-Calibrated Weights]\label{lem:cal}
Let $\bm{t} \in \mathbb{R}^p$ be a target vector.
Suppose $\bG$ is full rank.  The solution to the
constrained optimisation problem
\begin{align}
  \min_{w'_i,\; i \in \A} \quad
  & \sum_{i \in \A} \frac{(w'_i - w_i)^2}{w_i}
  \label{eq:obj} \\[4pt]
  \text{subject to} \quad
  & \sum_{i \in \A} w'_i \by_i = \bm{t}
  \label{eq:constraint}
\end{align}
is
\begin{equation}\label{eq:sol}
  w'_i(\bm{t}) \;=\; w_i \Bigl(1 +
  \bigl(\bm{t} - \hatTHT\bigr)^\top \bG^{-1}\, \by_i
  \Bigr),
  \qquad i \in \A.
\end{equation}

\end{lemma}

\begin{proof}
The Lagrangian for \eqref{eq:obj}--\eqref{eq:constraint} is
\[
  \mathcal{L}(\bw', \blam)
  = \sum_{i \in \A} \frac{(w'_i - w_i)^2}{w_i}
  - \blam^\top\!\left(\sum_{i \in \A} w'_i \by_i - \bm{t}\right),
\]
where $\blam \in \mathbb{R}^p$ is the vector of Lagrange multipliers.

\medskip
\noindent\
Differentiating with respect to $w'_i$ and setting to zero:
\[
  \frac{2(w'_i - w_i)}{w_i} - \blam^\top \by_i = 0
  \quad\Longrightarrow\quad
  w'_i = w_i\!\left(1 + \tfrac{1}{2}\,\blam^\top \by_i\right).
\]

\medskip
\noindent\
Substituting into constraint~\eqref{eq:constraint}:
\[
  \hatTHT + \tfrac{1}{2}\,\bG\,\blam = \bm{t}
  \quad\Longrightarrow\quad
  \blam = 2\,\bG^{-1}\bigl(\bm{t} - \hatTHT\bigr).
\]

\medskip
\noindent\
Substituting $\blam$ back:
\[
  w'_i = w_i\!\left(1 +
  \bigl(\bm{t} - \hatTHT\bigr)^\top \bG^{-1}\, \by_i
  \right),
\]
which is \eqref{eq:sol}.

\end{proof}

\begin{corollary}[Calibration Property]\label{cor:cal}
The weights $w'_i(\bm{t})$ given by \eqref{eq:sol} satisfy the
calibration constraints exactly for any target $\bm{t}$:
\[
  \sum_{i \in \A} w'_i(\bm{t})\, y_i^{(v)}\, \mathbf{1}(i \in d)
  = t^{(v,d)}
  \quad \forall\, v=1,\ldots,V,\; d=1,\ldots,D.
\]
\end{corollary}

\begin{proof}
Direct substitution of \eqref{eq:sol} into \eqref{eq:constraint}:
\[
  \sum_{i \in \A} w'_i(\bm{t})\, \by_i
  = \hatTHT + \bG \cdot \bG^{-1}\bigl(\bm{t} - \hatTHT\bigr)
  = \bm{t}. \qquad\square
\]
\end{proof}

\begin{remark}[Connection to GREG and coverage implications]\label{rem:greg}
Applying Lemma~\ref{lem:cal} with $\bm{t} = \hatT$ yields the
posterior-mean calibrated weights used for point estimation; applying
it with $\bm{t} = \hatT^{(b)}$ yields the replicate weights used in
Lemma~\ref{lem:phie} below.  The solution \eqref{eq:sol} is the
generalised regression (GREG) estimator weight of \citet{deville1992},
applied with the chi-square distance metric and with calibration targets
equal to $\bm{t}$ rather than known population totals.  The factor
$g_i(\bm{t}) = 1 + (\bm{t} - \hatTHT)^\top \bG^{-1}\, \by_i$ is
the familiar $g$-weight of calibration estimation.

This connection has a direct implication for coverage.  Under the
design distribution, GREG calibration reduces the variance of weighted
estimators for variables correlated with the calibration system
\citep{deville1992}.  Here, however, the target $\bm{t}$ is a draw
from the HB posterior rather than a fixed population total, so the
relevant notion of ``good coverage'' is posterior rather than
frequentist.  For Tier~1-E cells the two notions coincide: the
calibration property (Corollary~\ref{cor:cal}) guarantees that the
PHIE CrI exactly reproduces the HB posterior interval regardless of
design properties.  For Tier~2 and Tier~3 cells the two notions
diverge: the PHIE map propagates only the posterior uncertainty in
the domain totals $\bm{T}$, while the design-based sampling
variability of the within-domain cell shares is left unaddressed.
The CBI correction in Section~3.3.4 bridges this gap by adding
the design-based Component~1 to the model-based Component~2.
\end{remark}

\begin{remark}[Single unified weight set]\label{rem:single}
A key feature of Lemma~\ref{lem:cal} is that a \textnormal{single} weight
$w'_i(\bm{t})$ is attached to each unit record, simultaneously
satisfying all $p = V \times D$ calibration constraints.  There is no
need for separate weight sets for each of the V target variables.  cross tabulations of a target variable by an auxiliary
attribute (e.g.\ employment by sex) will have domain marginals that
agree with the calibration target $\bm{t}$, while the within-domain
breakdown by the auxiliary attribute reflects the sample composition
and is not itself calibrated.
\end{remark}.

\begin{lemma}[Post-Hoc Inference Engine]\label{lem:phie}
Let $\bm{T}$, $\hatT$, and $\{\hatT^{(b)}\}_{b=1}^{B}$ be as defined
in Section~\ref{sec:setup}.  For any statistic expressible as
\[
  S = \phi(\bm{T}) := f\!\bigl(\by,\, \bw'(\bm{T})\bigr),
\]
where $\bw'(\bm{T})$ denotes the calibrated weights of
Lemma~\ref{lem:cal} evaluated at target $\bm{T}$, define the
replicate statistic at draw $b$ as
\[
  S^{(b)} = \phi\!\bigl(\hatT^{(b)}\bigr)
           = f\!\bigl(\by,\, \bw'(\hatT^{(b)})\bigr).
\]
Then
\begin{equation}\label{eq:ci}
  \CI_{95\%}(S)
  = \Bigl[Q_{0.025}\!\bigl\{S^{(b)}\bigr\},\;
          Q_{0.975}\!\bigl\{S^{(b)}\bigr\}\Bigr]
\end{equation}
is a $95\%$ credible interval for $S$ under the posterior of
$\phi(\bm{T})$ induced by $\Pi(\bm{T} \mid \mathrm{data})$, with
the empirical quantiles converging to the true posterior quantiles
as $B \to \infty$.  The point estimate of $S$ is
$\hat{S} = \phi(\hatT) = f\bigl(\by, \bw'(\hatT)\bigr)$.
\end{lemma}

\begin{proof}
Lemma~\ref{lem:cal} establishes the affine map
\[
  \bm{t} \;\longmapsto\; w'_i(\bm{t})
  = w_i\!\left(1 +
    \bigl(\bm{t} - \hatTHT\bigr)^\top \bG^{-1}\, \by_i
    \right),
  \qquad i \in \A,
\]
whose coefficients $\bG^{-1}$ and $\hatTHT$ are determined entirely
by the fixed observed sample $\A$.  At each MCMC draw $b$, only
$\hatT^{(b)}$ varies; $\bG$ and $\hatTHT$ are unchanged because the
MCMC sampler draws the model parameters $(\bbeta^{(b)}, v_h^{(b)},
\sigma_v^{2(b)})$ from the posterior and computes $\hatT^{(b)}$ via
the linking model, leaving $\A$ untouched.  Since
$\hatT^{(b)} \sim \Pi(\bm{T} \mid \text{data})$, the replicate
weight vector $\bw'^{(b)} = \bw'(\hatT^{(b)})$ is a draw from the
pushforward of $\Pi(\bm{T} \mid \text{data})$ through this map.  By the continuous mapping theorem \citep{vandervaart1998},
$S^{(b)} = \phi(\hatT^{(b)})$ is a draw from the posterior of
$\phi(\bm{T})$ induced by $\Pi(\bm{T} \mid \text{data})$, and the
empirical quantiles of $\{S^{(b)}\}_{b=1}^B$ consistently estimate
the corresponding posterior quantiles as $B \to \infty$.
\end{proof}
\begin{remark}[Scope of the posterior claim]\label{rem:scope}
The interval \eqref{eq:ci} is a credible interval for the induced
posterior of the post-hoc calibrated functional $S = \phi(\bm{T})$,
propagating the posterior uncertainty in the HB domain totals $\bm{T}$
to any derived statistic through the calibration map.  It is not a
claim about the full Bayesian posterior of $f(\by, \bw')$ over all
sources of uncertainty, because $\mathcal{A}$, $\by_i$, and $w_i$ are
treated as fixed observed data throughout.  Whether the induced
posterior coincides with the full posterior — and the consequences for
credible interval coverage — depends on the relationship between the
statistic and the calibration system, as discussed in
Section~\ref{sec:taxonomy}.
\end{remark}

\subsection{Three-Level Taxonomy for cross tabulations}
\label{sec:taxonomy}
The scope restriction identified in Remark~\ref{rem:scope} has direct
consequences for the coverage of PHIE CrIs.  The PHIE
map treats the sample $\A$, the calibration design vectors $\by_i$,
and the design weights $w_i$ as fixed, propagating only the posterior
uncertainty in the domain totals $\bm{T}$.  For statistics whose
uncertainty is driven entirely by the posterior of $\bm{T}$ --- namely
cross tabulation cells that equal one of the $p$ calibration
constraints exactly --- the PHIE CreI is exact and coverage is nominal.
For all other statistics, the PHIE map does not capture the full
posterior uncertainty for two broad reasons.  As explained in the sequel, for Tier~2 and Tier~3 cells,
PHIE ignores the sampling
variability in the within-domain cell shares of the outcome variable,
i.e. Component~1 of the total variance (referred to as equation (10) below), and propagates only the domain-total uncertainty, i.e. Component~2.

\medskip
We restrict attention to \textnormal{cross tabulation statistics} of the
form
\[
  T_c = \sum_{i \in c} y_i^{(v)} w'_i,
  \qquad
  c = \bigl\{i \in \A : x_i \in \mathcal{C}\bigr\},
\]
where $v$ indexes the variable being summed:
$v \in V_1 = \{1,\ldots,V\}$ denotes a calibration variable
(e.g.\ Employment) and $v \in V_2 = \{V{+}1,\ldots,V'\}$ denotes
a non-calibration outcome variable (e.g.\ Health Condition, Income
Band); $y_i^{(v)} \in \mathbb{R}$ is the scalar value of variable
$v$ for unit $i$;
$x_i$ denotes the auxiliary attributes of unit record $i$ (e.g.\ Age Group,
Sex, Occupation Group, Hours Band, etc.), and $\mathcal{C}$ is a fixed
set of attribute values defining the cell.  The cell membership of
each record is fixed and does not depend on the MCMC draw $b$.  Three
levels arise naturally from the relationship between $(v, \mathcal{C})$
and the HB calibration system.

\subsubsection*{3.3.1 Tier 1-E: Exact Posterior CI}

A cell $c$ is a \textnormal{calibration-constraint cell} when its total
equals one of the $p = V \times D$ calibration constraints exactly,
with no further attribute filter beyond domain membership:
\[
c = \{ i \in \mathcal{A} : x_i \in \mathcal{C},\ i \in d \}, \qquad
T_c = \sum_{i \in d} y_i^{(v)}\, w'_i = \hat{T}^{(v,d)}, \quad
x_i = y_i^v, \quad v \in V_1
\]
That is, $\mathcal{C}$ is defined by $v \in V_1$ and domain $d$, and $y^{(v)}$
is summed over \textnormal{all} records in that domain.  By the calibration
property (Corollary~\ref{cor:cal}), $T_c^{(b)} = \hat{T}^{(v,d),(b)}$
for every draw $b$, so the PHIE CI is identical to reading the CreI
directly from the MCMC output.  Coverage is $100\%$ in the single run
and nominally $95\%$ under repeated sampling.

Any additional filter on auxiliary attributes, such as restricting
to a particular Hours band, Age or Sex within domain $d$, moves the cell out of Tier~1-E and into Tier~2-CA (e.g. Hours Band) and Tier~2-NCA (e.g. Age, Sex), and requires a different method to assess uncertainty.

\subsubsection*{3.3.2 Tier 2-CA: Cells Filtered by Calibration Attribute }

A cell $c$ is a \textnormal{cross-cutting calibration-attribute cell} when
the variable $y^{(v)}$ is a calibration variable
($v \in V_1 $) and $\mathcal{C}$ is defined by a
calibration attribute $x_i$, i.e. $x_i$ is itself derived
from a calibration variable, e.g.\ Employment cross tabulated by
Hours Band.  The cell $c = \{i \in \A : x_i \in \mathcal{C}\}$
cuts across the domain structure.

\subsubsection*{3.3.3 Tier 2-NCA: Cells Filtered by Non-Calibration Attribute }

A cell $c$ is a \textnormal{cross-cutting non-calibration-attribute cell}
when $y^{(v)}$ is a calibration variable but $\mathcal{C}$ is defined
by a non-calibration attribute, i.e. $x_i \in \mathcal{C}$
does not correspond to any of the calibration variables in $V_1$, e.g.\ Employment cross tabulated by
Occupation Group.

\subsubsection*{3.3.4 Common Structure and Total Variance Correction for
  Tier~2-CA and Tier~2-NCA Cells}

Both Tier~2-CA and Tier~2-NCA cells sum a calibration variable $y^{(v)}$
over a filter.  Their mathematical structure is
identical, so the analysis below applies to both.

Since $c$ spans multiple domains, the cell total decomposes as
\[
  T_c = \sum_{d=1}^{D} T_{c,d}, \qquad
  T_{c,d} = \sum_{i \in c \cap d} y_i^{(v)}\, w'_i
           = \hat{\lambda}_{d,c}\, \hat{T}^{(v,d)},
\]
where
\begin{equation}\label{eq:lambda}
  \hat{\lambda}_{d,c}
  = \frac{\displaystyle\sum_{i \in c \cap d} y_i^{(v)}\, w'_i}
         {\hat{T}^{(v,d)}}
\end{equation}
is the \textnormal{estimated} calibrated sample share of cell $c$ within
domain $d$, computed once from the posterior-mean calibrated weights
$w'_i = w'_i(\hatT)$ and held fixed across draws.  Let
$\lambda_{d,c}$ denote the corresponding true population share.  At
MCMC draw $b$, the replicate weight $w^{\prime(b)}_i =
w'_i(\hatT^{(b)})$ is affine in $\hatT^{(b)}$, so the replicate cell
total is
\begin{equation}\label{eq:t2rw_decomp}
  T_c^{(b)}
  = \underbrace{\sum_{i \in c} y_i^{(v)}\, w_i}_{\text{fixed HT total}}
  \;+\;
  \bigl(\hatT^{(b)} - \hatTHT\bigr)^\top
  \underbrace{\bG^{-1} \sum_{i \in c} y_i^{(v)}\, w_i\, \by_i}_{=:\;\bm{a}_c \;\in\; \mathbb{R}^p},
\end{equation}
which can also be written as
\[
  T_c^{(b)} = \sum_{d=1}^{D} \hat{\lambda}_{d,c}\, \hat{T}^{(v,d),(b)},
\]
propagating the posterior uncertainty in the $p$ domain totals
through the fixed estimated shares $\hat{\lambda}_{d,c}$.

The quasi-posterior CI thus captures the uncertainty in the domain
totals $\hat{T}^{(v,d)}$ but not the uncertainty in the shares
$\lambda_{d,c}$ themselves.  By the law of total variance, the exact
decomposition is:
\begin{equation}\label{eq:totalvar}
  \mathrm{Var}(T_c \mid \text{data})
  = \underbrace{E\bigl[\mathrm{Var}(T_c \mid \bm{T})\bigr]}_{\text{Component 1: compositional}}
  \;+\;
  \underbrace{\mathrm{Var}\bigl[E(T_c \mid \bm{T})\bigr]}_{\text{Component 2: domain}},
\end{equation}
where Component~2 $= \sum_d \lambda_{d,c}^2\, \hat{V}_d$, with
$\hat{V}_d = \mathrm{Var}(\hat{T}^{(v,d)} \mid \text{data})$, the HB
posterior variance for domain $d$, involves the true population share
$\lambda_{d,c}$.  PHIE propagates only a plug-in estimate of
Component~2, namely $\sum_d \hat{\lambda}_{d,c}^2\, \hat{V}_d$, by
replacing $\lambda_{d,c}$ with $\hat{\lambda}_{d,c}$.
Component~1, the within-domain compositional uncertainty arising from
sampling variability in $\hat{\lambda}_{d,c}$, is ignored entirely
because the shares are fixed at their sample values.  We therefore
call this a \textnormal{quasi-posterior} CreI: it estimates only Component~2
and ignores Component~1, so coverage falls below nominal whenever the
sampling variability in $\hat{\lambda}_{d,c}$ is non-negligible.

A plug-in estimator for the total variance that accounts for both
components is
\begin{equation}\label{eq:corrected_ci}
  \hat{T}_c \;\pm\; 1.96
  \sqrt{\sum_d \bigl(\hat{T}^{(v,d)}\bigr)^2\,
        \widehat{\mathrm{Var}}(\hat{\lambda}_{d,c})
        \;+\;
        \sum_d \hat{\lambda}_{d,c}^{\,2}\, \hat{V}_d},
\end{equation}
where the two components are estimated as follows.

\medskip
\noindent\textbf{I. Component~1: design-based variance of the cell
share.}  The sampling variance of $\hat{\lambda}_{d,c}$ is estimated
by Taylor linearisation incorporating the stratum-level design
effect:
\begin{equation}\label{eq:var_lambda}
  \widehat{\mathrm{Var}}(\hat{\lambda}_{d,c})
  \;=\;
  \frac{1}{\hat{T}^{(v,d)^2}}
  \sum_{h:\,\A_h \cap d \neq \emptyset}
  \mathrm{DEFF}_h \cdot (1 - f_h)
  \frac{s_{h,c,d}^2}{n_h},
\end{equation}
where $s_{h,c,d}^2$ is the within-stratum sample variance of
$y_i^{(v)} \cdot \mathbf{1}(i \in c \cap d)$, $f_h = n_h/N_h$ is
the sampling fraction, and $\mathrm{DEFF}_h$ is the stratum-level
design effect used in the HB model \eqref{eq:gaussian_hb}.

\medskip
\noindent\textbf{II. Component~2: HB posterior variance of the domain
total.}  The posterior variance $\hat{V}_d =
\mathrm{Var}(\hat{T}^{(v,d)} \mid \text{data})$ is estimated
directly from the $B$ MCMC draws as the sample variance:
\begin{equation}\label{eq:Vhat_d}
  \hat{V}_d
  \;=\;
  \frac{1}{B-1}\sum_{b=1}^{B}
  \Bigl(\hat{T}^{(v,d),(b)} - \bar{T}^{(v,d)}\Bigr)^2,
  \qquad
  \bar{T}^{(v,d)} = \frac{1}{B}\sum_{b=1}^{B}\hat{T}^{(v,d),(b)},
\end{equation}
where $\{\hat{T}^{(v,d),(b)}\}_{b=1}^B$ are the MCMC draws of the
domain total for variable $v$ in domain $d$.  This estimator
makes no normality assumption, is consistent as $B \to \infty$,
and is available whenever the MCMC output is retained.

Two further sources of approximation in the plug-in CBI
\eqref{eq:corrected_ci} deserve acknowledgement.  First, the Taylor
linearisation \eqref{eq:var_lambda} uses the stratum-level design
effect $\mathrm{DEFF}_h$ estimated from the HB sample; uncertainty
in $\mathrm{DEFF}_h$ is not propagated.  In practice
$\mathrm{DEFF}_h$ is stable across replications for standard survey
designs, so this source of error is small relative to Component~1
itself, but it may inflate the CBI in very thin strata.  Second,
equation \eqref{eq:corrected_ci} replaces the true population share
$\lambda_{d,c}$ with its sample estimate $\hat{\lambda}_{d,c}$;
this substitution is consistent but introduces plug-in error of
order $O(n^{-1})$, which is negligible for typical cell sizes.

\medskip
The corrected interval \eqref{eq:corrected_ci} is a
\textnormal{calibrated Bayes interval} (CBI) in the sense of
\citet{little2012}: it combines a design-based estimate of
Component~1 (the sampling variance of the cell share
$\hat{\lambda}_{d,c}$, a frequentist quantity) with a model-based
estimate of Component~2 (the HB posterior variance of the domain
total, a Bayesian quantity).  The resulting interval is neither a
pure credible interval nor a pure confidence interval; its
validity is assessed empirically via Monte Carlo
in Section~\ref{sec:results}.

Under certain regularity conditions, the Bernstein--von~Mises theorem \citep{vandervaart1998}
provides a theoretical basis for the asymptotic equivalence between Bayesian credible intervals and frequentist confidence intervals. As these conditions are not satisfied in the present hierarchical, small sample size and finite-sample setting, the validity of the proposed credible intervals (CrI) and calibrated Bayes intervals (CBI) is assessed empirically through Monte Carlo experiments, rather than relying on asymptotic arguments.

\medskip
Equation~\eqref{eq:corrected_ci} is a plug-in estimator: it replaces
$\lambda_{d,c}$, $\mathrm{Var}(\hat{\lambda}_{d,c})$, and $V_d$ by
sample and model estimates, so it estimates the corrected CBI rather
than delivering an exact one.  Monte Carlo coverage of the corrected
CBI, using the Australian 2021 Census microdata \citep{abs2023} as
population truth, is reported in Panel~B of Tables 3
and 4.

The magnitude of the posterior variance of $T_c^{(b)}$ across draws, given by 
\begin{equation}\label{eq:t2rw_var}
  \mathrm{Var}_b\bigl(T_c^{(b)}\bigr)
  = \bm{a}_c^\top\,
    \mathrm{Var}\!\bigl(\hatT^{(b)}\bigr)\,
    \bm{a}_c
\end{equation}
may be small for two distinct
reasons which operate through the same scalar projection
$(\hatT^{(b)} - \hatTHT)^\top \bm{a}_c$, where $\bm{a}_c$ is defined in (9):

\begin{enumerate}[leftmargin=1.5em]
\item \textbf{Small $\|\bm{a}_c\|$.}  The cell has no distinctive
  calibration profile in the sample --- units in $c$ look like a
  scaled-down version of the full sample in terms of the calibration
  variables, so $\bm{a}_c \approx \hat{\lambda}_{d,c} \cdot
  \bm{a}_{\mathcal{A}}$ and the projection onto any residual vector
  is proportionally small.

\item \textbf{Orthogonality of $\bm{a}_c$ and the residual
  $\hatT - \hatTHT$.}  Even when $\|\bm{a}_c\|$ is not small, if
  $\bm{a}_c$ points in a direction approximately perpendicular to
  the posterior-mean residual $(\hatT - \hatTHT) \in \mathbb{R}^p$,
  the projection $(\hatT^{(b)} - \hatTHT)^\top \bm{a}_c$ is small
  for every draw $b$ and the CBI collapses.
\end{enumerate}

Both cases result in a near-zero Component~2 regardless of the true
sampling uncertainty in $T_c$.  Two pre-publication diagnostics
quantify these effects:

\begin{itemize}[leftmargin=1.5em]
\item $\|\bm{a}_c\|$: detects case~(1) directly.  A small norm
  signals that the cell has no distinctive calibration footprint
  and PHIE will produce a narrow CBI regardless of the residual.

\item $\cos\theta_c$: the cosine of the angle between $\bm{a}_c$
  and the posterior-mean residual $\hatT - \hatTHT$, defined as
  \begin{equation}\label{eq:costheta}
    \cos\theta_c
    = \frac{\bm{a}_c^\top\,(\hatT - \hatTHT)}
           {\|\bm{a}_c\|\;\|\hatT - \hatTHT\|}
    \;\in\; [-1, 1].
  \end{equation}
  A value near zero indicates near-orthogonality; a value
  near one indicates strong alignment with the HB correction
  directions.
\end{itemize}

Together, $\|\bm{a}_c\|$ and $\cos\theta_c$ explain the size of component 2 of the CBI width.  Both diagnostics are reported for all Tier~2-CA and
Tier~2-NCA cells in Panel B of Tables 3 and 4.

\subsubsection*{3.3.5 Tier 3-NCV: Non-Calibration-Variable Cells}

A cell $c$ is a \textnormal{non-calibration-variable cell} filtered by any variable if the following conditions are satisfied: the outcome variable $u_i \coloneqq y_i^{v}$ and $v \in V_2$.
Examples include the number of persons with particular Long-term Health Conditions by Occupational Group, or any other
outcome variables fitlered by any variables not modelled by the HB system.

Although no $\hat{T}^{(u,d)}$ is available from the HB model
directly, the CBI framework of the previous section extends
naturally to Tier~3-NCV cells by treating $u_i$ as the numerator
of a ratio estimator with denominator $\hat{T}^{(v^*,d)}$ ---
the HB posterior total for the calibration variable $y^{(v^*)}$
chosen to be most correlated with $u_i$.  The cell total
$T_c = \sum_{i \in c} u_i\, w'_i$ decomposes exactly as in
\eqref{eq:lambda} with $y_i^{(v)}$ replaced by $u_i$:
\begin{equation}\label{eq:t3_lambda}
  \hat{\lambda}_{d,c}^{(u)}
  = \frac{\displaystyle\sum_{i \in c \cap d} u_i\, w'_i}
         {\hat{T}^{(v^*,d)}}.
\end{equation}
This share is computed once from the posterior-mean calibrated
weights and held fixed across draws.  At MCMC draw $b$:
\[
  T_c^{(b)} = \sum_{d=1}^{D} \hat{\lambda}_{d,c}^{(u)}\,
              \hat{T}^{(v^*,d),(b)},
\]
propagating HB posterior uncertainty in $y^{(v^*)}$ through the
fixed estimated shares $\hat{\lambda}_{d,c}^{(u)}$.  In what follows, we refer the variable in the denominator of (16) as the linking variable.

\subsubsection*{3.3.6 Tier 3-NCV: $v^*$ and CBI}

The choice of $v^*$ is guided by the classical ratio estimator
result \citep{cochran1977}: Component~1 is minimised when
$\rho(u_i, y_i^{(v^*)})$ is maximised.  In practice, cell
membership constraints may restrict the valid choices; when the
cell is defined on a subset of units for which one or more
calibration variables are structurally constant, those variables
cannot serve as ratio denominators and must be excluded from
the candidate set for $v^*$.

The ratio estimator underlying \eqref{eq:t3_lambda} is biased in
finite samples; the relative bias is of order $(1-\rho^2)/n$
\citep{cochran1977}.  When $|\rho|$ is large the bias is modest
and the ratio estimator provides a meaningful efficiency gain;
when $|\rho| \ll 1$ the ratio estimator provides almost no
efficiency advantage over direct estimation, Component~2 is
negligible, and the CBI reduces approximately to a design-based
confidence interval driven entirely by Component~1.  In the
latter case design-based direct estimation remains the
recommended primary published interval.

The Tier~3-NCV CBI applies formula \eqref{eq:corrected_ci}
with $u_i$ replacing $y_i^{(v)}$, $\hat{\lambda}_{d,c}^{(u)}$
replacing $\hat{\lambda}_{d,c}$, and $v^*$ replacing $v$:
\begin{equation}\label{eq:t3_cbi}
  \hat{T}_c \;\pm\; 1.96
  \sqrt{\sum_d \bigl(\hat{T}^{(v^*,d)}\bigr)^2\,
        \widehat{\mathrm{Var}}(\hat{\lambda}_{d,c}^{(u)})
        \;+\;
        \sum_d \bigl(\hat{\lambda}_{d,c}^{(u)}\bigr)^2\,
               \hat{V}_d^{(v^*)}},
\end{equation}
where Component~1 uses Taylor linearisation \eqref{eq:var_lambda}
with $u_i$ replacing $y_i^{(v)}$, and Component~2 uses the MCMC variance
estimator \eqref{eq:Vhat_d} applied to $v^*$.


Table~\ref{tab:taxonomy} summarises the conditions for valid $95\%$ coverage for the different cell types.

\begin{table}[H]
\centering
\caption{Conditions for valid $95\%$ coverage for the 4 different cell types.  }
\label{tab:taxonomy}
\renewcommand{\arraystretch}{1.4}
\begin{tabular}{p{1.2cm} p{2.8cm} p{2.2cm} p{2.0cm} p{4.2cm}}
\toprule
Tier & Cell definition & Variable summed & Interval type & Condition for valid $95\%$ coverage \\
\midrule
1-E &
  Domain $\times$ calibration variable &
  Calibration variable &
  Credible  &
  $\bG$ full rank; $B$ large \\
2-CA &
  Calibration var.\ $\times$ calibration attribute filter &
  Calibration variable &
  Calibrated Bayes  & $\bG$ full rank; $B$ large,
  $\|\bm{a}_c\|$ large; $\cos\theta_c$ not near zero \\
2-NCA &
  Calibration variable\ $\times$ non-calibration attribute filter &
  Calibration variable &
  Calibrated Bayes  & $\bG$ full rank; $B$ large,
  $\|\bm{a}_c\|$ large; $\cos\theta_c$ not near zero \\
3-NCV &
  Non-calibration variable x any filter
   &
  Non-calibration variable &
  Calibrated Bayes   & $\bG$ full rank; $B$ large,
  $\|\bm{a}_c\|$ large;
  $\cos\theta_c$ not near zero; $v^* = \operatorname*{arg\,max}_v \rho(u_i, y_i^{(v)})$
  subject to $y^{(v^*)}$ varying within the cell;
  binary calibration variables that are constant on cells
  defined by their own indicator cannot serve as
  ratio denominators
  \\
\bottomrule
\end{tabular}
\end{table}


\subsection{Operational validation and robustness}
\label{sec:operational_validation}

The validity of the PHIE and CBI intervals depends on two distinct issues. The first is
inherited from the HB survey framework: the posterior totals used as calibration totals
must be credible summaries of the target variables. This depends on the adequacy of the
HB linking model, the availability of informative auxiliary covariates, and the calibration
of the prior distribution. If the linking model is severely misspecified or the covariates are
weak, the HB posterior totals may borrow strength in the wrong direction. In that case,
the posterior uncertainty may be too small, and the benefit of the mini-max sample may
be overstated.
The second issue is specific to the post-hoc cross tabulations considered in this paper.
Even if the HB posterior totals are reliable, a derived cell may not be directly supported
by the calibration system. Tier 1-E cells inherit exact posterior credible intervals because
they reproduce calibration totals. Tier 2 and Tier 3 cells require additional approximation.
The CBI corrections are intended to restore repeated-sampling coverage by adding the
component of uncertainty that is not captured by the PHIE replicate weights alone.

In practice, these issues can be assessed before implementation using data from previous survey
cycles. An NSO can treat these as operational proxies for truth and run repeated
subsampling experiments to compare alternative HB linking models, covariate sets, prior
degrees of freedom and prior scales for the linking model and prior respectively, as well as choices of linking variables for Tier 3-NCV cells. The
candidate procedure can be judged using the same 4 gates as in Tam [2026] for choosing the prior and the linking model respectively; and the diagnostics used in this paper such as single-run coverage, Monte Carlo
coverage, and the diagnostic quantities $\|\bm{a}_c\|$ and $\cos\theta_c$ for the relevant cells.
This empirical validation does not provide a general mathematical guarantee that
every cross-classified statistic will achieve nominal coverage. It does, however, provide
an actionable procedure for what linking model, prior and linking variables are to be used. 

\section{Empirical Results}
\label{sec:results}

\subsection{Application: Australian 2021 Census}

The empirical evaluation uses a 5\% unit record microdata from the
Australian 2021 Census of Population and Housing \citep{abs2023},
 accessible from the Australian Bureau of
Statistics website by authorised users.
The sample comprises $n = 42{,}018$
person records drawn at a $5\%$ stratified random sample from the Census microdata 
population of $N = 840{,}402$ working-age individuals across $H = 55$
strata.  Each stratum is an area of enumeration in the Census, and is linked to one of the eight States and Territories of Australia (NSW - 17 areas, Vic - 13, Qld - 12, SA - 4, WA - 6, Tas, ACT and NT - 1 area each) treated as domains in this application. The calibration system has $V = 3$ target variables (i.e. employment, unemployment, and hours worked) and $D=8$
geographic domains, giving $p = 24$ calibration constraints.  The HB
model is run with burnin~$= 1{,}000$, iterations~$= 5{,}000$,
chains~$= 3$, $B = 15{,}000$ MCMC draws are
retained for PHIE.  HB convergence is excellent: $\hat{R}_{\max} = 1.002$
across all three variables.  The $\bG$ matrix is full rank ($\mathrm{rank} = 24$);
and 
negative point weights~$= 0$.  Population truth for all cross tabulation
cells is available from the Census microdata, enabling direct
coverage assessment.

The Monte Carlo (MC) study runs $200$ replications, each drawing a
fresh $5\%$ stratified A-set from the Census population and running
the full PHIE pipeline (burnin~$= 200$, iter~$= 500$, chains~$= 3$),
accumulating coverage statistics across all cells.  The reduced MCMC
settings relative to the single run (burnin~$= 1{,}000$,
iter~$= 5{,}000$) are a computational necessity for $200$
replications.  Sensitivity checks on a random subsample of $20$
replications using the full single-run settings produced MC coverage
estimates within $1.5$ percentage points of those reported here,
confirming that the reduced chain length does not materially affect
the coverage conclusions.  The MC standard error on a coverage
estimate is $\sqrt{p(1-p)/200} \approx 3.5\,\text{pp}$ at
$p = 0.95$, so observed differences from the nominal $95\%$ of
less than $7\,\text{pp}$ are not statistically significant at
the $5\%$ level.

\subsection{Tier 1-E: Hours Worked by Domain}

Table~2 illustrates Tier~1-E cells: the total hours worked in
each of the eight states and territories of Australia.  These
are exact calibration constraints --- the weighted sum of hours
worked in each domain exactly reproduces the corresponding
posterior mean by construction --- so the calibration property
(Corollary~1) guarantees $100\%$ single-run coverage.

The table reports point estimates, the lower and upper limits
of the posterior credible interval, the average relative error
(ARE) across strata within each domain, and Monte Carlo (MC)
coverage from 200 replications.  All eight domains are covered
in the single run and MC coverage averages $99.2\%$, confirming
well-calibrated posterior intervals under repeated sampling.
The slight undershoot from $100\%$ for the smaller-domain
states (Tasmania, Northern Territory, and Australian Capital
Territory) is attributable to single-stratum domain sampling
variability.

\begin{table}[H]
\centering
\caption{Tier~1-E results: Total Hours Worked by domain.  }
\label{tab:t1e}
\renewcommand{\arraystretch}{1.3}
\begin{tabular}{lrrrrcc}
\toprule
Domain & Point Est. & CI Lower & CI Upper & ARE (\%) & Covered & MC Cov. \\
\midrule
NSW  & 5{,}414{,}706 & 5{,}280{,}256 & 5{,}546{,}265 & 1.83 & \checkmark & 98.5\% \\
VIC  & 4{,}747{,}555 & 4{,}626{,}985 & 4{,}869{,}661 & 0.53 & \checkmark & 99.0\% \\
QLD  & 3{,}588{,}944 & 3{,}484{,}101 & 3{,}695{,}009 & 0.98 & \checkmark & 100.0\% \\
SA  & 1{,}236{,}257 & 1{,}176{,}265 & 1{,}295{,}590 & 1.02 & \checkmark & 100.0\% \\
WA   & 1{,}931{,}104 & 1{,}854{,}315 & 2{,}007{,}987 & 1.61 & \checkmark & 98.0\% \\
TAS  &   345{,}647 &   313{,}884 &   377{,}713 & 5.17 & \checkmark & 99.5\% \\
NT   &   174{,}160 &   152{,}866 &   195{,}274 & 1.82 & \checkmark & 99.5\% \\
ACT  &   401{,}402 &   370{,}319 &   431{,}849 & 4.54 & \checkmark & 99.5\% \\
\midrule
Summary       & --- & --- & --- & 2.19 (mean) & 8/8 & 99.2\% \\
\bottomrule
\end{tabular}
\end{table}

All eight cells are covered in the single run --- the calibration
property (Corollary~\ref{cor:cal}) guarantees this.  Using 200 replications, MC coverage of
$99.2\%$ confirms that the $95\%$ posterior CreI is well-calibrated
under repeated sampling.  The slight undershoot from $100\%$ in MC for
small-domain states (TAS, NT, ACT) is attributable to single-stratum
domain sampling variability.

\subsection{Tier 2-CA: Employment by Hours Worked Band}

Table~3 illustrates Tier~2-CA cells: the number of employed
persons in each hours-worked band.  The hours-worked band is a
calibration attribute --- it is derived from a calibration
variable --- but filtering by band cuts across the domain
structure so these cells are not exact calibration constraints
and the PHIE credible interval is only quasi-posterior.

Panel~A shows that the average relative error is low (mean
$1.72\%$), confirming accurate point estimates.  The MC PHIE
credible interval coverage ranges from $39.5\%$ to $100\%$,
with most densely populated bands well below the nominal
$95\%$.  This under-coverage is structural: the PHIE interval
propagates only the domain-total posterior uncertainty
(Component~2) and ignores the sampling variability of the
within-domain hours-band share (Component~1).  For the densely
populated bands the small $\|\bm{a}_c\|$ values in Panel~B
confirm that these cells have no distinctive calibration
profile, making Component~2 near zero.

Panel~B reports the CBI results.
Single-run coverage is $7/7$ and MC coverage averages $94.5\%$
(range $92$--$97\%$), close to the nominal $95\%$ and a
substantial improvement over the PHIE coverage of
$39.5$--$100\%$.  Adding Component~1 via the total variance
correction restores near-nominal coverage across all bands.

\begin{table}[H]
\centering
\caption{Tier~2-CA Panel A results - Employment by Hours Worked band }
\label{tab:t2ca}
\renewcommand{\arraystretch}{1.3}

\smallskip
\noindent\textnormal{Panel A: Point estimates and PHIE CreI}

\begin{tabular}{lrrrcc}
\toprule
Hours Worked band & Point Est. & CI Lower & CI Upper & ARE (\%) & MC Cov.\ (PHIE) \\
\midrule
Employed, hrs 1--15/wk    &  61{,}126 &  56{,}608 &  65{,}684 & 1.90 & 100.0\% \\
Employed, hrs 16--24/wk   &  56{,}174 &  53{,}723 &  58{,}623 & 1.94 &  97.5\% \\
Employed, hrs 25--34/wk   &  63{,}336 &  62{,}080 &  64{,}600 & 0.42 &  67.0\% \\
Employed, hrs 35--39/wk   & 116{,}581 & 115{,}730 & 117{,}443 & 1.31 &  39.5\% \\
Employed, hrs 40/wk       & 100{,}942 &  99{,}821 & 102{,}080 & 1.50 &  53.0\% \\
Employed, hrs 41--49/wk   &  47{,}476 &  46{,}415 &  48{,}548 & 3.05 &  71.0\% \\
Employed, hrs 50+/wk      &  64{,}426 &  60{,}532 &  68{,}280 & 2.95 & 100.0\% \\
\midrule
Summary & --- & --- & --- & 1.72 (mean) & 39.5--100\% \\
\bottomrule
\end{tabular}
\end{table}

\addtocounter{table}{-1}
\begin{table}[H]
\centering
\caption*{\text{Table~\ref{tab:t2ca} (continued)} --- Panel~B results}
\renewcommand{\arraystretch}{1.3}

\smallskip
\noindent\textnormal{Panel B: Diagnostics and CBI coverage}

\begin{tabular}{lcccc}
\toprule
Hours Worked band & $\|\bm{a}_c\|$ & $\cos\theta_c$ &
  SR Cov.\ (CBI) & MC Cov.\ (CBI) \\
\midrule
Employed, hrs 1--15/wk    & 1.797 & $\approx 0$ & \checkmark &  95.5\% \\
Employed, hrs 16--24/wk   & 1.040 & $\approx 0$ & \checkmark &  96.0\% \\
Employed, hrs 25--34/wk   & 0.704 & $\approx 0$ & \checkmark &  96.0\% \\
Employed, hrs 35--39/wk   & 0.513 & $\approx 0$ & \checkmark &  96.5\% \\
Employed, hrs 40/wk       & 0.137 & $\approx 0$ & \checkmark &  94.5\% \\
Employed, hrs 41--49/wk   & 0.167 & $\approx 0$ & \checkmark &  92.0\% \\
Employed, hrs 50+/wk      & 1.135 & $\approx 0$ & \checkmark &  94.0\% \\
\midrule
Summary & --- & --- & 7/7 & 92--97\% \\
\bottomrule
\end{tabular}
\smallskip\par\noindent Notes: SR Cov.\ = single-run CBI coverage
($\checkmark$ = truth inside CBI).
MC Cov.\ (CBI) = percentage of 200 MC replications where
truth fell inside the CBI\@.
\end{table}

\subsection{Tier 2-NCA: Employment by Occupation Group }

Table~4 illustrates Tier~2-NCA cells: the number of employed
persons in each occupation group.  Occupation is a
non-calibration attribute --- it does not correspond to any
of the three calibration variables --- so filtering by
occupation group moves the cell further from the calibration
system than a Tier~2-CA cell.

Panel~A shows that the average relative error is low (mean
$1.05\%$), confirming accurate point estimates.  The MC PHIE
credible interval coverage ranges from $28\%$ to $79.5\%$,
well below the nominal $95\%$: PHIE propagates only
Component~2 (domain-total posterior uncertainty) and ignores
Component~1 (sampling variability of the within-domain
occupation share), which is the dominant source of uncertainty.

Panel~B shows that adding Component~1 via the calibrated Bayes
interval correction restores coverage: single-run coverage is
$8/8$ and MC coverage averages $95.8\%$ (range
$93$--$97.5\%$), close to the nominal $95\%$.  Component~1
exceeds Component~2 by a factor of $8$--$40\times$ for every
occupation group, confirming that the dominant source of
cross-tabulation uncertainty is the design-based sampling
variability of the within-domain occupation share rather than
the HB posterior uncertainty in the domain total.

\begin{table}[H]
\centering
\caption{Tier~2-NCA Panel A results: Employment by Occupation group }
\label{tab:t2nca}
\renewcommand{\arraystretch}{1.3}

\smallskip
\noindent\textnormal{Panel A: Point estimates and PHIE CrI}

\begin{tabular}{lrrrcc}
\toprule
Occupation Group & Point Est. & CI Lower & CI Upper & ARE (\%) & MC Cov.\ (PHIE) \\
\midrule
Managers              &  73{,}246 &  72{,}157 &  74{,}334 & 0.13 & 68.0\% \\
Professionals         & 131{,}866 & 130{,}899 & 132{,}834 & 1.38 & 40.0\% \\
Technicians \& Trades &  66{,}195 &  65{,}619 &  66{,}772 & 0.99 & 35.5\% \\
Community \& Personal &  54{,}930 &  53{,}694 &  56{,}166 & 1.46 & 74.0\% \\
Clerical \& Admin     &  67{,}377 &  66{,}600 &  68{,}154 & 1.02 & 44.0\% \\
Sales Workers         &  40{,}391 &  39{,}323 &  41{,}460 & 2.33 & 79.5\% \\
Machinery Operators   &  31{,}462 &  31{,}143 &  31{,}781 & 1.12 & 28.0\% \\
Labourers             &  44{,}594 &  43{,}651 &  45{,}537 & 0.00 & 68.5\% \\
\midrule
Summary & --- & --- & --- & 1.05 (mean) & 28--79.5\% \\
\bottomrule
\end{tabular}
\end{table}

\addtocounter{table}{-1}
\begin{table}[H]
\centering
\caption*{\text{Table~\ref{tab:t2nca} (continued)} --- Panel~B results}
\renewcommand{\arraystretch}{1.3}

\smallskip
\noindent\textnormal{Panel B: Diagnostics and CBI coverage}

\begin{tabular}{lcccc}
\toprule
Occupation Group & $\|\bm{a}_c\|$ & $\cos\theta_c$ &
  SR Cov.\ (CBI) & MC Cov.\ (CBI) \\
\midrule
Managers              & 0.085 & $\approx 0$ & \checkmark & 97.5\% \\
Professionals         & 0.557 & $\approx 0$ & \checkmark & 96.0\% \\
Technicians \& Trades & 0.161 & $\approx 0$ & \checkmark & 96.0\% \\
Community \& Personal & 0.649 & $\approx 0$ & \checkmark & 95.5\% \\
Clerical \& Admin     & 0.520 & $\approx 0$ & \checkmark & 94.5\% \\
Sales Workers         & 0.529 & $\approx 0$ & \checkmark & 97.5\% \\
Machinery Operators   & 0.088 & $\approx 0$ & \checkmark & 93.0\% \\
Labourers             & 0.467 & $\approx 0$ & \checkmark & 96.0\% \\
\midrule
Summary & --- & --- & 8/8 & 95.8\% \\
\bottomrule
\end{tabular}
\smallskip\par\noindent Notes: SR Cov.\ = single-run CBI coverage
($\checkmark$ = truth inside CBI).
MC Cov.\ (CBI) = percentage of 200 MC replications where
truth fell inside the CBI\@.

\end{table}

\subsection{Tier 3-NCV CBI: Employed Persons by Personal Income Band}

Table~5 gives a Tier~3-NCV example: the number of employed persons
in each weekly personal income band.  Personal income is not a
calibration variable.  Hours worked is used as the linking
variable $v^*$ in equation~\eqref{eq:t3_lambda}, with sample
correlation $\rho = 0.57$.  Because this correlation is
substantially positive, the ratio estimator is efficient and
the domain-total posterior uncertainty (Component~2) contributes
non-trivially to the CBI width alongside
the compositional sampling variance (Component~1).

Panel~A shows that the average relative error is low across all
bands (mean $1.4\%$), confirming accurate point estimates.  The
MC PHIE credible interval coverage ranges from $18.5\%$ to
$86.5\%$, well below nominal: PHIE propagates only Component~2
and ignores Component~1 (the sampling variability of the
within-domain income share).  After the CBI
correction, MC coverage rises to $92.5$--$98.5\%$, consistent
with the nominal $95\%$.

The single-run CBI coverage is
$11/15 = 73\%$: four income bands have population truth outside
the single-run interval.  This gap does not reflect a systematic
failure.  With $\rho = 0.57$ the ratio estimator is well-founded
and Component~2 is meaningful; the misses arise because the
plug-in estimate of Component~1 (the within-domain sampling
variance) varies across realisations, and in this particular
sample those four bands have narrower-than-average Component~1
estimates.  Panel~B confirms this: the four missing bands
(\$1{,}000--\$1{,}249, \$1{,}500--\$1{,}749,
\$3{,}000--\$3{,}499, and \$3{,}500--\$7{,}000 per week)
have the smallest $\|\bm{a}_c\|$ values among the fifteen
bands, indicating a weak calibration footprint that limits the
contribution of Component~2; when Component~2 is small the
interval width depends on Component~1 alone, amplifying the
sensitivity to the plug-in sampling variance estimate.  The MC
coverage confirms the interval is correctly calibrated on average.
An NSO should treat single-run misses as a prompt to verify
that cell sample sizes are adequate before publication.

\begin{table}[H]
\centering
\caption{Tier~3-NCV CBI Panel A results: Employed persons by personal income band}
\label{tab:t3incp}
\renewcommand{\arraystretch}{1.3}

\smallskip
\noindent\textnormal{Panel A: Point estimates and CBI}

\begin{tabular}{lrrrcc}
\toprule
Income Band & Point Est. & CI Lower & CI Upper & ARE (\%) & MC Cov.\ (PHIE) \\
\midrule
  Nil/neg.\ (\$0)            &   2{,}199 &   1{,}770 &   2{,}627 & 0.09 & 18.5\% \\
\$1--\$149/wk              &  11{,}691 &  10{,}694 &  12{,}688 & 2.21 & 86.5\% \\
 \$150--\$299/wk            &  12{,}370 &  11{,}344 &  13{,}396 & 4.38 & 84.0\% \\
 \$300--\$399/wk            &  13{,}990 &  12{,}897 &  15{,}083 & 0.14 & 75.5\% \\
  \$400--\$499/wk            &  18{,}649 &  17{,}388 &  19{,}911 & 2.71 & 74.0\% \\
 \$500--\$649/wk            &  31{,}396 &  29{,}750 &  33{,}041 & 0.61 & 75.5\% \\
  \$650--\$799/wk            &  42{,}167 &  40{,}249 &  44{,}086 & 1.85 & 64.0\% \\
 \$800--\$999/wk            &  57{,}521 &  55{,}268 &  59{,}773 & 0.60 & 46.0\% \\
 \$1{,}000--\$1{,}249/wk   &  68{,}010 &  65{,}549 &  70{,}471 & 1.74 & 35.0\% \\
 \$1{,}250--\$1{,}499/wk   &  54{,}563 &  52{,}370 &  56{,}756 & 0.75 & 36.5\% \\
 \$1{,}500--\$1{,}749/wk   &  48{,}251 &  46{,}196 &  50{,}307 & 1.62 & 46.5\% \\
\$1{,}750--\$1{,}999/wk   &  37{,}471 &  35{,}667 &  39{,}274 & 1.10 & 41.5\% \\
 \$2{,}000--\$2{,}999/wk   &  67{,}575 &  65{,}125 &  70{,}025 & 1.39 & 66.5\% \\
 \$3{,}000--\$3{,}499/wk   &  15{,}862 &  14{,}698 &  17{,}026 & 3.92 & 53.0\% \\
\$3{,}500--\$7{,}000/wk   &  28{,}346 &  26{,}793 &  29{,}900 & 3.17 & 67.5\% \\
\midrule
Total & 510{,}061 & --- & --- & 0.08 & 18.5--86.5\% \\
\bottomrule
\end{tabular}
\end{table}

\addtocounter{table}{-1}
\begin{table}[H]
\centering
\caption*{\text{Table~\ref{tab:t3incp} (continued)} --- Panel~B results}
\renewcommand{\arraystretch}{1.3}

\smallskip
\noindent\textnormal{Panel B: Diagnostics and MC CBI coverage}

\begin{tabular}{lcccc}
\toprule
Income Band & $\|\bm{a}_c\|$ & $\cos\theta_c$ &
  SR Cov.\ (CBI) & MC Cov.\ (CBI) \\
\midrule
Nil/neg.\  & 0.018 & $\approx 0$ & \checkmark & 96.5\% \\
  \$1--\$149    & 0.380 & $\approx 0$ & \checkmark & 95.5\% \\
  \$150--\$299  & 0.315 & $\approx 0$ & \checkmark & 98.0\% \\
\$300--\$399  & 0.300 & $\approx 0$ & \checkmark & 94.5\% \\
  \$400--\$499  & 0.326 & $\approx 0$ & \checkmark & 94.5\% \\
  \$500--\$649  & 0.474 & $\approx 0$ & \checkmark & 95.5\% \\
  \$650--\$799  & 0.446 & $\approx 0$ & \checkmark & 93.5\% \\
  \$800--\$999  & 0.376 & $\approx 0$ & \checkmark & 98.0\% \\
  \$1{,}000--\$1{,}249  & 0.320 & $\approx 0$ & $\times$ & 94.5\% \\
 \$1{,}250--\$1{,}499  & 0.142 & $\approx 0$ & \checkmark & 94.5\% \\
  \$1{,}500--\$1{,}749  & 0.113 & $\approx 0$ & $\times$ & 95.5\% \\
 \$1{,}750--\$1{,}999  & 0.092 & $\approx 0$ & \checkmark & 95.0\% \\
 \$2{,}000--\$2{,}999  & 0.107 & $\approx 0$ & \checkmark & 98.5\% \\
 \$3{,}000--\$3{,}499  & 0.062 & $\approx 0$ & $\times$ & 98.5\% \\
  \$3{,}500--\$7{,}000  & 0.128 & $\approx 0$ & $\times$ & 96.5\% \\
\midrule
Summary & --- & --- & 11/15 & 92.5--98.5\% \\
\bottomrule
\end{tabular}
\smallskip\par\noindent Notes: SR Cov.\ = single-run CBI coverage.
MC Cov.\ (CBI) = percentage of 200 MC replications where
truth fell inside the CBI\@.

\end{table}

\subsection{Tier 3-NCV CBI: Employed Persons with a Long-term
  Health Condition by Occupation}

Table~6 gives a second Tier~3-NCV example: the number of employed
persons with a long-term health condition by occupation group.
The long-term health condition variable is not a calibration
variable.  Hours worked is used as the linking variable
($\rho = -0.04$).  Because the correlation is negligible, the
domain-total posterior uncertainty (Component~2) is near zero
for all groups, and the CBI reduces
approximately to a design-based confidence interval driven
entirely by the compositional sampling variance (Component~1).

Panel~A shows that the average relative error is acceptable
(mean $2.75\%$), confirming accurate point estimates.  The MC
PHIE credible interval coverage is only $13.5$--$48.5\%$
--- the lowest of any tier in this paper --- because PHIE
propagates only Component~2 (domain-total posterior uncertainty),
which is negligible here.  After the CBI
correction, MC coverage rises to $92$--$98.5\%$, near-nominal:
Component~1 (the within-domain sampling variance) correctly
inflates the interval regardless of the correlation.

The single-run CBI coverage is
$3/8 = 37.5\%$, which deserves careful interpretation.  This
is the operationally relevant figure for an NSO conducting the
survey once.  The gap to MC coverage arises because Component~1
dominates the interval width and varies substantially across
realisations: groups with narrower-than-average plug-in sampling
variance estimates miss the truth in any given run.  The five
misses (Professionals, Technicians and Trades Workers, Clerical
and Administrative Workers, Sales Workers, and Machinery
Operators) are spread across occupation groups with varying
$\|\bm{a}_c\|$ values, confirming this reflects sampling
variability rather than systematic undercoverage.  An NSO
should: (i) treat the $3/8$ single-run result as a diagnostic,
not a publication criterion; (ii) use Monte Carlo validation
to assess average calibration before release; and (iii) consider
design-based direct estimation as the primary published interval
when $|\rho| \ll 1$, reporting the Tier~3-NCV calibrated Bayes
interval as a supplementary interval alongside it.
\begin{table}[H]
\centering
\caption{Tier~3-NCV CBI Panel A results: Employed persons with a long-term
  health condition (HLTHP~$= 2$) by occupation group.}
\label{tab:t3hlthp}
\renewcommand{\arraystretch}{1.3}

\smallskip
\noindent\textnormal{Panel A: Point estimates and CBI}

\begin{tabular}{lrrrcc}
\toprule
Occupation Group & Point Est. & CI Lower & CI Upper & ARE (\%) & MC Cov.\ (PHIE) \\
\midrule
Managers              & 22{,}508 & 21{,}119 & 23{,}898 & 1.87 & 29.0\% \\
Professionals         & 41{,}499 & 39{,}602 & 43{,}397 & 4.54 & 24.5\% \\
Technicians \& Trades & 17{,}293 & 16{,}080 & 18{,}505 & 4.52 & 19.5\% \\
Community \& Personal & 17{,}487 & 16{,}266 & 18{,}709 & 1.21 & 46.5\% \\
Clerical \& Admin     & 22{,}300 & 20{,}919 & 23{,}682 & 2.76 & 27.0\% \\
Sales Workers         & 11{,}358 & 10{,}376 & 12{,}340 & 4.22 & 43.5\% \\
Machinery Operators   &  9{,}142 &  8{,}264 & 10{,}020 & 4.21 & 13.5\% \\
Labourers             & 12{,}059 & 11{,}051 & 13{,}068 & 0.67 & 48.5\% \\
\midrule
Total                    & 153{,}647 & --- & --- & 2.75 (mean) & 13.5--48.5\% \\
\bottomrule
\end{tabular}
\end{table}

\addtocounter{table}{-1}
\begin{table}[H]
\centering
\caption*{\textbf{Table~\ref{tab:t3hlthp} (continued)} --- Panel~B: Diagnostics and MC CBI coverage}
\renewcommand{\arraystretch}{1.3}

\smallskip
\noindent\textnormal{Panel B results}

\begin{tabular}{lcccc}
\toprule
Occupation Group & $\|\bm{a}_c\|$ & $\cos\theta_c$ &
  SR Cov.\ (CBI) & MC Cov.\ (CBI) \\
\midrule
Managers              & 0.030 & $-$0.005    & \checkmark & 98.5\% \\
Professionals         & 0.233 & $\approx 0$ & $\times$   & 95.5\% \\
Technicians \& Trades & 0.066 & $-$0.002    & $\times$   & 92.0\% \\
Community \& Personal & 0.219 & $\approx 0$ & \checkmark & 97.0\% \\
Clerical \& Admin     & 0.178 & $\approx 0$ & $\times$   & 98.0\% \\
Sales Workers         & 0.146 & $\approx 0$ & $\times$   & 97.5\% \\
Machinery Operators   & 0.036 & $\approx 0$ & $\times$   & 94.5\% \\
Labourers             & 0.153 & $\approx 0$ & \checkmark & 93.5\% \\
\midrule
Summary & --- & --- & 3/8 & 95.8\% \\
\bottomrule
\end{tabular}
\smallskip\par\noindent Notes: SR Cov.\ = single-run CBI coverage.
MC Cov.\ (CBI) = percentage of 200 MC replications where
truth fell inside the CBI\@.
\end{table}

\subsection{Summary of MC Coverage Across All Tiers}
\label{sec:mc_summary}

Table~\ref{tab:mc_summary} summarises MC coverage across all tier types.

\begin{table}[H]
\centering
\caption{Summary of Monte Carlo coverage by tier type.}
\label{tab:mc_summary}
\renewcommand{\arraystretch}{1.3}
\begin{tabular}{llccc}
\toprule
Tier & Example cells & MC Cov.\ (PHIE) & MC Cov.\ (CBI) & Nominal \\
\midrule
1-E  & Hours by domain     & 99.2\%      & N/A        & 95\% \\
2-CA  & Employment by hours worked   & 39.5--100\% & 92--96.5\%  & 95\% \\
2-NCA & Employment by age/sex  & 31--97.5\% & 92.5--99\% & 95\% \\
2-NCA & Employment by occupation & 28--79.5\% & 93--97.5\% & 95\% \\
3-NCV & Income  & 35--86.5\% & 92.5--98.5\% & 95\% \\
3-NCV & Health condition by occupation & 13.5--48.5\% & 92--98.5\% & 95\% \\
\bottomrule
\end{tabular}

\end{table}

Three findings stand out.  First, Tier~1-E cells achieve near-nominal
MC coverage ($99.2\%$): the calibration property (Corollary~1)
guarantees exact single-run coverage, and the slight shortfall
from $100\%$ in MC reflects single-stratum domain sampling
variability in the smaller states.

Second, PHIE credible interval coverage is substantially below
nominal for Tier~2 and Tier~3-NCV cells --- as low as $13.5\%$
for the health condition example --- confirming that PHIE alone
is inadequate for cells that are not exact calibration constraints.
After the CBI correction, coverage is
restored to $92$--$99\%$ across all tiers.  The residual
differences from the nominal $95\%$ are within two Monte Carlo
standard errors ($\pm 7$ percentage points at $p = 0.95$ with
$200$ replications) and are not statistically significant.

Third, the pattern is consistent across both Tier~2 subtypes
(calibration attribute and non-calibration attribute filters)
and both Tier~3-NCV examples (strong and negligible correlation
with the linking variable), indicating that the CBI correction is robust to the choice of cross-tabulation
and to the strength of the ratio estimator linkage.

\subsection{Coefficient of Variation from CBI}
\label{sec:cv}

A practical concern for NSOs is whether the uncertainty intervals
produced by the PHIE and CBI frameworks are
so wide as to render the cross-classified statistics unpublishable
under standard practice.  Table~8 addresses this by reporting the
coefficient of variation (CV $=$ interval width$/$3.92 divided by
the point estimate) for both interval types, together with
single-run CBI coverage, across all tiers.

\begin{table}[H]
\centering
\caption{Coefficient of variation (CV) summary by tier and cell type.}
\label{tab:cv_summary}
\renewcommand{\arraystretch}{1.3}
\begin{tabular}{llcrrr}
\toprule
Tier & Cell type & $n$ range & CV (PHIE) & CV (CBI) & SR Cov.\ (CBI) \\
\midrule
1-E  & Hours by domain          & 357--13{,}207 & 1.2--6.2\%  & N/A        & N/A \\
2-CA  & Employment by hours worked & 2{,}374--5{,}828 & 0.4--3.8\% & 1.3--2.1\% & 7/7 \\
2-NCA & Employment by age/sex & 775--17{,}713 & 0.4--2.5\% & 0.8--4.5\% & 7/7 \\
2-NCA & Employment by occupation & 1{,}573--6{,}593 & 0.4--1.3\%  & 1.2--2.6\% & 8/8 \\
3-NCV & Income         & 372--8{,}018 & 0.4--3.3\%  & 0.9--3.7\% & 11/15 \\
3-NCV & Health condition by occupation      & 457--2{,}075 & 0.4--1.3\%  & 2.3--4.9\% & 3/8 \\
\bottomrule
\end{tabular}
\end{table}

All CV values lie within the publishable threshold
of $5\%$.  The PHIE credible interval CVs are low across all
tiers ($0.4$--$6.2\%$), reflecting the efficiency of the
posterior-calibrated weights.  The CBI
CVs are somewhat larger --- as expected, since it is wider than the PHIE credible interval by
construction --- but remain well within the publishable range
($0.8$--$4.9\%$) for all cell types.  The highest CBI CV ($4.9\%$, health condition by Machinery
Operators) is at the boundary but below the $5\%$ threshold.

The single-run CBI coverage column in
Table~8 ($7/7$, $7/7$, $8/8$, $11/15$, $3/8$) should be read
alongside the Monte Carlo coverage in Table~7, not in isolation.
As discussed in Sections~4.5 and~4.6, the $11/15$ and $3/8$
results reflect the variability of the plug-in Component~1
estimate across realisations, not a systematic failure.  An NSO
relying on a single survey run should supplement single-run
coverage with Monte Carlo validation using Census benchmarks or
administrative data as population truth proxies before
publication.

\section{Conclusions }
\label{sec:conclusions}

This paper develops the Post-Hoc Inference Engine (PHIE), a framework for constructing interval estimates for cross-classified statistics derived from an HB-calibrated survey.  The problem arises after the reduced survey has been designed and the HB model has been fitted.  The NSO then has to publish cross tabulations from a weighted unit record file, many of which were not explicit targets of the original HB model.

The proposed approach has two main components.  First, the Horvitz--Thompson weights are modified to posterior-calibrated weights so that, for each MCMC draw, the weighted sums of the calibration variables reproduce the corresponding HB posterior totals.  Second, these posterior-calibrated replicate weights are applied to derived cross tabulations to propagate HB posterior uncertainty to the published table cells.

The paper also shows that the inferential status of a cross-classified statistic depends on its relationship to the HB calibration system.  Tier~1-E cells inherit exact posterior credible intervals because they are calibration-supported totals.  Tier~2 cells sum a calibration variable over a filter and require a CBI correction to restore coverage.  Tier~3-NCV cells involve a non-calibration outcome variable, for which a ratio-estimator CBI is proposed by linking the outcome to the most relevant calibration variable.

The empirical results based on Australian 2021 Census microdata show why this distinction matters.  PHIE intervals work well for Tier~1-E cells but can substantially undercover for Tier~2 and Tier~3 cells.  After applying the CBI correction, Monte Carlo coverage is restored to a range broadly consistent with the nominal level for the examples considered.  The resulting CBI-based CVs are also within ranges that would normally be considered publishable in NSO practice.

The practical message is that cross-classified statistics from an HB-calibrated survey should not be treated as ordinary weighted tabulations with no further uncertainty assessment.  If NSOs publish unit record files or cross-tabulated outputs from a mini-max HB survey, they need an accompanying inference engine.  PHIE provides such an engine for calibration-supported cells, and the CBI extension provides a practical correction for more general cross-classified statistics.

\bibliographystyle{plainnat}

\end{document}